\documentclass{article}
\usepackage{graphics,latexsym,color,fullpage}
\usepackage{sgame}
\usepackage{amsmath, amssymb, amsthm}
\usepackage{tikz}
\usetikzlibrary{positioning,shadows,shapes,backgrounds,fit}

\usepackage{hyperref}
\hypersetup{colorlinks=true,linkcolor=blue,citecolor=blue,filecolor=blue,urlcolor=blue}

\theoremstyle{plain}
\newtheorem{theorem}{Theorem}
\newtheorem{lemma}[theorem]{Lemma}
\newtheorem{obs}[theorem]{Observation}
\newtheorem{cor}[theorem]{Corollary}

\theoremstyle{definition}
\newtheorem{definition}{Definition}

\theoremstyle{remark}
\newtheorem{remark}{Remark}
\newtheorem{example}{Example}

\newcommand{\Prob}[2]{\Pr_{#1} \left( #2 \right)}

\newcommand{\tv}[1]{\left\|#1\right\|}

\newcommand{\1}{{\mathbf 1}}
\newcommand{\s}{{\mathbf s}}

\renewcommand{\epsilon}{\varepsilon}

\definecolor{emphcol}{gray}{.75}
\newcommand{\pne}[1]{
\colorbox{emphcol}{{
          #1
         }
   }
}

\begin{document}

\title{Imperfect best-response mechanisms}
\author{
Diodato~Ferraioli\thanks{DIAG, Sapienza Universit\`a di Roma}
\and
Paolo~Penna\thanks{Institute of Theoretical Computer Science, ETH Zurich}}
\date{}

\maketitle

\begin{abstract}
Best-response mechanisms (Nisan, Schapira, Valiant, Zohar, 2011) provide a unifying framework for studying various distributed protocols in which the participants are instructed to repeatedly best respond to each others' strategies. Two fundamental features of these mechanisms are convergence and incentive compatibility.

This work investigates convergence and incentive compatibility conditions of such mechanisms when players are not guaranteed to always best respond but they rather play an \emph{imperfect} best-response strategy. That is, at every time step every player deviates from the prescribed best-response strategy according to some probability parameter. The results explain to what extent convergence and incentive compatibility depend on the assumption that players never make mistakes, and how robust such protocols are to ``noise'' or ``mistakes''.
\end{abstract}

\raggedbottom

\section{Introduction}
One of the key issues in designing a distributed protocol (algorithm) is its convergence to a stable state, also known as self-stabilization. Intuitively, starting from any initial (arbitrarily corrupted) state, the protocol should eventually converge to  the ``correct state'' as intended by the designer.
Incentive compatibility considerations have been also become important in the study of distributed protocols since the participants cannot be assumed to altruistically implement the protocol if that is not beneficial for themselves.

A unifying game-theoretic approach for proving  
both convergence and incentive compatibility has been recently proposed by Nisan et al. \cite{NSVZ11}. They consider so-called  \emph{best-response mechanisms} or \emph{dynamics} in which the protocol prescribes  that each participant (or player) should simply best-respond to the strategy currently played by the other players. Essentially the same base \emph{game} is played over and over (or until some equilibrium is reached), with players updating their strategies in \emph{some} (unspecified) order. Nisan et al. \cite{NSVZ11} proved that for a suitable class of games the following happens:
\begin{itemize}
 \item \emph{Convergence.} The dynamics eventually
reaches a unique equilibrium point (a unique pure Nash equilibrium) of the base game regardless of the order in which players respond (including  concurrent responses).
 \item \emph{Incentive compatibility.}  A player who deviates from the prescribed best-response strategy can only worsen his/her final utility, that is, the dynamics will reach a different state that yields weakly smaller payoff.
\end{itemize}
These two conditions say that the protocol will eventually ``stabilize'' if implemented correctly, and that the participants are actually willing to do so.
Convergence itself is a rather strong condition because no assumption is made on how players are scheduled for updating their strategies, a  typical situation in \emph{asynchronous} settings. Incentive compatibility is also non-trivial because a best-response is a \emph{myopic} strategy which does not take into account the future updates of the other players. In fact, neither of these conditions can be guaranteed on general games.

Nisan et al. \cite{NSVZ11} showed that several protocols and mechanisms arising in computerized and economics settings are in fact best-response mechanisms over the \emph{restricted} class of games for which convergence and incentive compatibility are always guaranteed. Their applications include: (1) the Border Gateway Protocol (BGP) currently used in the Internet, (2) a game-theoretic version of the TCP protocol, and (3) mechanisms for the classical cost-sharing and stable roommates problems studied in micro-economics.

\medskip \noindent
In this work we address the following question:
\begin{quote}
\emph{What happens to these protocols/mechanisms if players do not always best respond?}
\end{quote}
Is it possible that when players sometimes deviate from the prescribed protocol (e.g., by making mistakes in computing their best-response or by scarce knowledge about other players' actions) then the protocol does not converge anymore? Can such mistakes induce some other player to adopt a non-best-response strategy that results in a better payoff? Such questions arise naturally from fault tolerant considerations in protocol design, and have several connections to equilibria computation and bounded rationality issues in game theory.

\paragraph{Our contribution.}
We investigate convergence and incentive compatibility conditions of mechanisms (dynamics) described in \cite{NSVZ11} when players are not guaranteed to always best respond but they rather play an \emph{imperfect} best-response strategy. That is, at every time step every player deviates from the prescribed best-response strategy according to some probability parameter $p\geq 0$. The parameter $p$ can be regarded as the probability of making a mistake
every time the player updates his/her strategy.

Our results indicate to what extent convergence and incentive compatibility depend on the assumption that players never make mistakes, and provide necessary and sufficient conditions for  the robustness of these mechanisms/dynamics:
\begin{itemize}
 \item \emph{Convergence.} Because of mistakes convergence can be achieved only in a probabilistic sense.
We give bounds on the parameter $p$ in order to guarantee convergence with sufficiently good probability.

One might think that for small values of $p$ our dynamics behaves (approximately) as the dynamics without mistakes, i.e. it converges to an equilibrium point regardless of the order in which players respond. However, it turns out this is not the case. Indeed,
our first negative result (Theorem~\ref{th:negative}) shows that even when  $p$ is exponentially small in the number $n$ of players the dynamics does not converge, i.e., the probability of being in the equilibrium is always small (interestingly, such negative result applies also to certain instances of BGP in the realistic model of Gao and Rexford \cite{GR01}).

The proof of this result shows the existence of a particularly ``bad'' schedule that amplifies the probability that the imperfect dynamics deviates from the perfect one. This highlights that imperfect dynamics differ from their perfect counterpart in which convergence results \emph{must consider} how players are scheduled. Indeed, we complement the negative result above with a general positive result (Theorem~\ref{th:convergence_solvable}) saying that convergence
can be guaranteed whenever $p$ is polynomially small in
some parameters defining both the game and the schedule of the players.
For such values of $p$, the upper bound on the convergence time of dynamics without mistakes is (nearly) an upper bound for the \emph{imperfect} best-response dynamics.
\item \emph{Incentive compatibility.} We first observe that games that are incentive compatible for dynamics without mistakes, may no longer be incentive compatible for imperfect best-response dynamics (Theorem~\ref{th:ic-negative}). In other words, a player who deviates incidentally from the mechanism induces another player to \emph{deliberately} deviate.  A sufficient condition for incentive compatibility of imperfect best-response mechanisms (Theorem~\ref{th:p-BR:incentive-compatible}) turns out to be a quantitative version of the one given in \cite{NSVZ11}. Roughly speaking, if the payoffs of the Nash equilibrium are \emph{sufficiently} larger than the other possible payoffs, then incentive compatibility holds. As the probability $p$ of making mistakes vanishes, the class of games for which convergence and incentive compatibility holds tends to the class of games in \cite{NSVZ11}.

\end{itemize}
Our focus is on the same class of (base) games of \cite{NSVZ11} since this is the only known  general class for which best-response dynamics converge (regardless of the schedule of the players) and are incentive compatible. In our view this class is important as it describes accurately certain protocols that are implemented in practice and it unifies several results in game theory. In particular, the mathematical model of how the commercial relationships between Autonomous Systems  (the Gao-Rexford model \cite{GR01}) leads to games in this class and, ultimately, to the fact that BGP converges and is incentive compatible \cite{LSZ11,NSVZ11}. Considering more general games for the analysis of BGP would in fact produce ``wrong'' results (constructing unrealistic examples for which the protocol does not converge or is not incentive compatible).

We nevertheless take one step further and apply the tools from \cite{NSVZ11} (and this work) to a natural generalization of their games. Intuitively speaking, these games guarantee only that best-response converge to a \emph{subgame}. In this case, the dynamics of the original game can be approximated by the dynamics of the subgame (Theorem~\ref{th:reducible}). Unfortunately, this ``reduction'' cannot be pushed further simply because the subgame can be an arbitrary game and different $p$-imperfect best-response dynamics lead to different equilibria (even for the same $p$; see Section~\ref{subsec:impossible}). However, when the dynamics on the subgame are well-understood, then we can infer their behavior also on the original game.

\paragraph{Related work.}
Convergence of best-response dynamics is a main topic in game theory. It relates to the so-called problem of equilibrium selection, that is, how the players converge to an equilibrium (see the book by Harsanyi and  Selten \cite{HarSel88}). Noisy versions of such dynamics have been studied in order to consider the effects of bounded rationality and limited knowledge of the players (which limits their ability to compute their best responses).

Our imperfect best response dynamics are similar to the \emph{mutation model} by Kandori et al.~\cite{KMR93}, and to the \emph{mistakes model} by Young~\cite{Young93}, and Kandori and Rob~\cite{KR95}. A related model is the \emph{logit dynamics} of Blume~\cite{blumeGEB93} in which the probability of a mistake depends on the payoffs of the game. All of these works assume a particular schedule of the players (the order in which they play in the dynamics). Whether such an assumption effects the selected equilibrium is the main focus of a recent work  by Alos-Ferrer and Netzer~\cite{afnGEB10}. They studied convergence of these dynamics on general games when the parameter $p$ vanishes, and provide a characterization of the resulting equilibria in terms of a kind of potential function of the game.
Convergence results that take into account non-vanishing $p$ are only known for fixed dynamics on specific class of games (see e.g. \cite{blumeGEB93,AulFerPasPenPer11,AFPPG12}).

Incentive compatibility of best-response dynamics provide a theoretical justification for several protocols and auctions widely adopted in practice. Levin et al.~\cite{LSZ11} proved convergence and incentive compatibility of the intricate BGP protocol in the current Internet (based on the mathematical model by Gao and Rexford~\cite{GR01} that captures the commercial structure that underlies the Internet and explains convergence of BGP). The theoretical analysis of TCP-inspired games by Godfrey et al.~\cite{GSZS10}
shows that certain variants of the current TCP protocol converge (the flow rate stabilizes) and are incentive compatible on arbitrary networks (this property assumes routers adopt specific queuing policy). The so-called Generalized Second-Price auctions used in many ad-auctions is another example of incentive compatible best-response mechanism as proved by Nisan et al.~\cite{NSVZ11auctions}. All of these problems (and others) and results have been unified by Nisan et al.~\cite{NSVZ11} in their framework.

\section{Definitions}\label{sec:defs}
We consider an $n$-player base game $G$ in which  each player $i$ has a finite set of strategies $S_i$, and a utility function $u_i$. Each player can select a strategy $s_i \in S_i$ and the vector
$\s=(s_1,\ldots,s_n)$ is the corresponding strategy profile, with $u_i(\s)$ being the payoff of player $i$.
To stress the dependency of the utility $u_i$ on the strategy $z_i$ of player $i$ we adopt the standard notation $(z_i,\s_{-i})$ to denote the vector $(s_1,\ldots,s_{i-1},z_i,s_{i+1},\ldots, s_n)$. 

\begin{definition}[best response]\label{def:best-response}
 A strategy $s_i^*\in S_i$ is a \emph{best response} to the strategy profile $\s_{-i}$ if this strategy maximizes $i$'s payoff, that is, for every other strategy $s_i' \in S_i$ it holds
\[
 u_i(s_i',\s_{-i}) \leq u_i(s_i^*,\s_{-i}).
\]
\end{definition}

\paragraph{(Imperfect) Best-response dynamics.}
A \emph{game dynamics} consists of a (possibly infinite) sequence of strategy profiles $\s^0,\s^1,\ldots$, where $\s^0$ is an arbitrarily chosen profile and the profile $\s^t$ is obtained from $\s^{t-1}$ by letting some of the players updating their strategies.  Therefore a game dynamics is determined by a \emph{schedule} of the players specifying, for each time step, the subset of players that are selected for updating their strategies and a \emph{response rule}, which specifies how a player updates her strategy (possibly depending on the past history and on the current strategy profile). A natural and well-studied response rule prescribes that players (myopically) choose the best response to the strategies currently chosen by the others (according to Definition~\ref{def:best-response}).

In this work we focus on dynamics based on the following kind of schedules and response rules.
As for the response rule, we consider a scenario in which a selected player can deviates from the (prescribed) best-response.
\begin{definition}[$p$-imperfect response rule]
 A response rule is $p$-imperfect if a player does not update her strategy to the best-response with probability at most $p$.
\end{definition}

Examples of these rules are given in the mutation \cite{KMR93} or mistakes models \cite{Young93,KR95} (see Appendix~\ref{apx:mistakes} for a brief overview). 
The best-response rule is obviously $0$-imperfect, which we also denote as \emph{perfect}.
The response rule in logit dynamics \cite{blumeGEB93} (see Appendix~\ref{apx:logit} for a brief overview) is $p$-imperfect with
\[
 p \leq \frac{m-1}{m-1 + e^\beta}
\]
for all games in which the payoff are integers\footnote{When the minimum difference in the payoff of an agent between a best-response and a non-best response is $\gamma$, this extends easily by taking $\beta_\gamma= \beta \cdot \gamma$ in place of $\beta$.} and each player has at most $m$ strategies.

In order to avoid trivial impossibility results on convergence we need to consider a non-adaptive adversarial schedule that satisfies some reasonable fairness condition.
We allow both deterministic and randomized schedules satisfying the following definition.
\begin{definition}[$(R,\varepsilon)$-fair schedule]\label{def:fair-selection}
 A schedule is $(R,\varepsilon)$-fair if there exists a nonnegative integer $R$ such that, for any interval of $R$ time steps, all players are selected at least once in this interval with probability at least $1-\varepsilon$, i.e. for every player $i$ and any time step $a$ we have
 $$
  \Pr(SEL_{i,a,R}) \geq 1- \varepsilon,
 $$
 where $SEL_{i,a,R}$ is the event that player $i$ is selected at least once in the interval $[a+1,a+R]$.
\end{definition}
The convergence rate of a dynamics is measured in number of rounds, where a \emph{round} is a sequence of consecutive time steps in which each player is selected for update at least once. The definition above affects the (expected) length of a round and the probability that a round has a certain length (number of steps). In particular, if a schedule is $(R,\varepsilon)$-fair, then, for every $0 < \delta < \varepsilon$, all players are selected at least once with probability at least $1 - \delta$ in an interval of $R \cdot \left\lceil \frac{\log (1/\delta)}{\log(1/\varepsilon)}\right\rceil$ time steps (this holds because the probability $1-\varepsilon$ is guaranteed for any interval of $R$ time steps). Another important parameter of the schedule is the \emph{maximum} number of players selected for update in one step by the schedule, denoted as $\eta$ (with $\eta \leq n$).

\begin{example}
 Scheduling players in round-robin fashion or concurrently corresponds to $(n,0)$-fair and $(1,0)$-fair schedules, respectively. Selecting a player uniformly at random at each time step is $(R,\varepsilon)$-fair with $R = O(n \log n)$.
\end{example}

Henceforth, we always refer as \emph{imperfect best-response dynamics} to any dynamics whose schedule is $(R,\varepsilon)$-fair and whose response rule is $p$-imperfect, and as \emph{imperfect best-response mechanisms} to the class of all imperfect best-response dynamics. 
Note that we do not put any other constraint on the way the dynamics run.  
For instance, the response rule of the players \emph{can} depend on the history of the dynamics, as long as the probability of not choosing a best response is at most $p$. 

\paragraph{Convergence and incentive compatibility.} We define convergence and incentive compatibility as in Nisan et al. \cite{NSVZ11}.
We say that a game dynamics for a game $G$ \emph{converges} if it eventually converges to a (pure) \emph{Nash equilibrium} of the game, i.e. there exists $t > 0$ such that the strategy profile of players at time step $t$ coincides with a Nash equilibrium of $G$.
This definition might seem limited in a setting in which players can ``make mistakes'' since, once a Nash equilibrium is reached, there is a positive probability of leaving again from this state. However, we choose to adopt this definition of convergence (the same as in \cite{NSVZ11}) because when the parameters of the dynamics guarantee convergence to a Nash equilibrium in a ``small'' number of rounds with ``good probability'', then the dynamics is likely to remain in the Nash equilibrium for ``many'' steps (see Corollary~\ref{cor:nash}). This is useful in applications in which there is a ``termination condition'' (for instance, the dynamics ends if no player updates her strategy for a  certain number of consecutive steps).

The above considerations on ``termination condition'' leads to the following definition of \emph{total utility} and the resulting \emph{incentive compatibility} condition, both defined as in \cite{NSVZ11}. 
Let us denote with $X^t$ the random variable that represent the strategy profile induced by a dynamics on a game $G$ after $t$ time steps. If the dynamics terminates after some finite number of steps $T$, then the \emph{total utility} of a player $i$ is defined as $\Gamma_i = E\left[u_i\left(X^T\right)\right]$; otherwise, that is if the dynamics does not terminate after finite time, the total utility is defined as $\Gamma_i = \lim \sup_{t \rightarrow \infty} E\left[u_i\left(X^t\right)\right]$. Then, a dynamics for a game $G$ is \emph{incentive compatible} if playing this dynamics is a pure Nash equilibrium in a new game $G^\star$ in which players' strategies are all possible response rule that may be used in $G$ and
players' utilities are given by their total utilities. That is, a dynamics for a game $G$ is incentive compatible if every player does not improves her total utilities by playing according to a response rule different from the one prescribed, given that each other player does not deviate from the prescribed response rule.

\paragraph{Never best-response and the main result in \cite{NSVZ11}.}
Nisan et al.~\cite{NSVZ11} analyzed the convergence and incentive compatibility of the (perfect) best-response dynamics. Before stating their result, let us now recall some definitions.
\begin{definition}[never best-response]
A strategy $s_i$ is a never best-response (NBR) for player $i$ if, for every $\s_{-i}$, there exists $s_i'$ such that $u_i(s_i,\s_{-i}) < u_i(s_i',\s_{-i})$.\footnote{Nisan et al.~\cite{NSVZ11} assume that each player has also a \emph{tie breaking rule} $\prec_i$, i.e., a total order on $S_i$, that depends solely on the player's private information. In the case that a tie breaking rule $\prec_i$ has been defined for player $i$, then $s_i$ is a NBR for $i$ also if $u_i(s_i, \s_{-i}) = u_i(s_i', \s_{-i})$ and $s_i \prec_i s_i'$. However, such tie-breaking rule can be implemented in a game by means of suitable perturbations of the utility function: with such an implementation our definition become equivalent to the one given in \cite{NSVZ11}.}
\end{definition}
Note that according to a $p$-imperfect response rule, a player updates her strategy to a NBR with probability at most $p$.
\begin{definition}[elimination sequence]
An elimination sequence for a game $G$ consists of a sequence of subgames
$$
 G = G_0 \supset G_1 \supset \cdots \supset G_r=\hat G,
$$
where any game $G_{k+1}$ is obtained from the previous one by letting
a player $i^{(k)}$ eliminate strategies which are NBR in $G_{k}$.
\end{definition}
The length of the shortest elimination sequence for a game $G$ is denoted with $\ell_G$ (we omit the subscript when it is clear from the context).
It is easy to see that for each game $\ell_G \leq n (m-1)$, where $m$ is the maximum number of strategies of a player.

Our results will focus on the following classes of games.
\begin{definition}[NBR-reducible and NBR-solvable games]
\hyphenation{re-du-ci-ble}
The game $G$ is NBR-reducible to $\hat G$ if there exists an elimination sequence for $G$ that ends in $\hat G$. The game $G$ is NBR-solvable if it is NBR-reducible to $\hat G$ and $\hat G$ consists of a unique profile.
\end{definition}

\begin{example}
Consider a $2$-player game with strategies $\{0,1,2\}$ and the following utilities:
\begin{equation}\label{eq:game}
\begin{game}{4}{3}
      & 0     & 1 & 2 \\
0   &0,0    &0,0   & 0,-2  \\
1   &0,0    &-1,-1 & -1,-2   \\
2   &-2,0   &-2,-1 & -2,-2
\end{game}\hspace*{\fill}
\end{equation}
Notice that strategy $2$ is a NBR for both players. Hence, there exists an elimination sequence of length 2 that reduces above game in its upper-left $2 \times 2$ subgame with strategy set $\{0,1\}$ for each player. Therefore, this game is NBR-reducible. If we modify the utilities in this upper-left $2 \times 2$ subgame as follows
$$
\begin{game}{3}{2}
    & 0     & 1\\
0   &0,0    &0,$-\delta$\\
1   &$-\delta$,0    &-1,-1
\end{game}
$$
then the game reduces further to the profile $(0,0)$ and hence it is NBR-solvable. Observe that the unique profile at which the game $G$ is reduced in an NBR-solvable game is also the unique Nash equilibrium of the original game.
\end{example}

While the convergence result of \cite{NSVZ11} holds for the class of NBR-solvable games, in order to guarantee incentive compatibility they introduce the following condition on the payoffs:
\begin{definition}[NBR-solvable with clear outcome]\label{def:clear-outcome}
A NBR-solvable game is said to have a clear outcome if, for every player $i$, there is a player-specific elimination sequence such that the following holds. If $i$ appears the first time in this sequence at position $k$, then in the subgame $G_k$ the profile that maximizes the utility of player $i=i^{(k)}$ is the Nash equilibrium.
\end{definition}

\begin{theorem}[main result of \cite{NSVZ11}]
\label{th:main:Nisan-et-al}
Best-response dynamics of every NBR-solvable game $G$ converge to a pure Nash equilibrium of the game and, if $G$ has clear outcome, are incentive compatible. Moreover, convergence is guaranteed in $\ell_G$ rounds for any schedule, where a round is a sequence of consecutive time steps in which each player is selected for update at least once.
\end{theorem}

Note that convergence and incentive compatibility holds regardless of the schedule of players. Moreover, the theorem implies that for a specific $(R,\varepsilon)$-fair schedule the dynamics converges in $O(R \cdot \ell_G)$ time steps. Note also that convergence does \emph{not} require a clear outcome and this condition is only needed for incentive compatibility.

\section{Convergence properties}
\subsection{A negative result}
In this section we will show that the result about convergence of the best-response dynamics in NBR-solvable games given in~\cite{NSVZ11} is not resistant to the introduction of ``noise'', i.e.,
there is a NBR-solvable games and an imperfect best-response dynamics that never converges, except with small probability, to the Nash equilibrium even for values of $p$ very small.
Specifically we will prove the following theorem.
\begin{theorem}\label{th:negative}
For every $0 < \delta < 1$, there exists a $n$-player NBR-solvable game $G$ and an imperfect best-response dynamics with parameter $p$ exponentially small in $n$ such that for every integer $t > 0$ the dynamics converges after $t$ steps with probability at most $\delta$.
\end{theorem}
We highlight that this theorem does not state just that the probability of being in the Nash equilibrium at any time $t$ is low, but that the probability that the the system has ever been in the Nash equilibrium at any time before $t$ is low.

\paragraph{The game.}
Consider the following game: there are $n$ players with two strategies $0$ and $1$. For each player $i$, we define the utility function of $i$ as follows: if players $1, \ldots, i-1$ are playing $1$, then $i$ prefers to play $1$ regardless of the strategies played by players $i+1, \ldots, n$, i.e., $u_i(\1_{1 \cdots i-1},0,\s_{i+1 \cdots n}) < u_i(\1_{1 \cdots i-1},1,\s_{i+1 \cdots n})$ for each strategy profile $\s$, otherwise $i$ prefers to play $0$, i.e. $u_i(\s_{-i},1) < u_i(\s_{-i},0)$ for any $\s$ such that $\s_{1 \cdots i-1} \neq \1_{1 \cdots i-1}$.

It is easy to check that the above game is NBR-solvable. Indeed, the elimination sequence consists of players $1,2,\ldots, n$ eliminating strategy $0$ one-by-one in this order (note that $1$ is a dominant strategy for player $1$ and, more in general, strategy $1$ is dominant for $i$ in the subgame in which all players $1,\ldots, i-1$ have eliminated $0$). The subgame $\hat G$ consists of the unique pure Nash equilibrium that is the  profile $\mathbf 1 = (1,\ldots, 1)$.

\paragraph{The $p$-imperfect response rule.}
All players play the following $p$-imperfect response rule:
\begin{itemize}
 \item Player $i$ chooses strategy $0$ with probability $p$ if all players $j<i$ are playing strategy $1$;
 \item Player $i$ chooses strategy $0$ with probability $1 - q$ if at least one player $j<i$ is playing strategy $0$, where $0 \leq q \ll p$.
\end{itemize}

\paragraph{The $(2^{n-1},0)$-fair schedule.}
Let us start by defining sequences $\sigma_i$, with $i = 1, \ldots, n$, recursively as follows
$$
 \sigma_1=1, \quad  \sigma_2 = 12, \quad
 \sigma_3 = 1213, \quad  \ldots \quad
 \sigma_i = \sigma_{i-1} \sigma_{i-2} \cdots \sigma_1 i.
$$
Observe that each sequence has length $2^{i-1}$. Then players are scheduled one at a time according to $\sigma_n$ and then repeat.
A key observation about the schedule is in order.
\begin{obs}\label{obs:adversary}
 Between any two  occurrences of player $i < n$ there is an occurrence of a player $j\geq i+1$.
\end{obs}
Intuitively speaking, this property causes any \emph{bad move} of some player in the sequence $\sigma_n$ to propagate  to the last player $n$, where by ``bad move'' we mean that at time $t$ the corresponding player $\sigma_n(t)$ plays strategy $0$ given that each player $j < \sigma_n(t)$ plays $1$ (thus, a bad move occurs with probability $p$).

\begin{proof}[Proof (of Theorem~\ref{th:negative})]
Throughout the proof we will denote $2^{n-1}$ as $\tau$ for sake of readability.

Let $X_t$ be the random variable that represents the profile of the game at step $t$. We will denote with $X_t^n$ the $n$-th coordinate of $X_t$, i.e. the strategy played by player $n$ at time $t$. Suppose that player $n$ plays $0$ at the beginning. Then, for every $t < \tau$, the probability that at time step $t$ the game is in a Nash equilibrium is obviously $0$. Consider now $t \geq \tau$. The probability that at time step $t$ the game is in a Nash equilibrium is obviously less than the probability that $X^n_t = 1$. Hence it will be sufficient to show that $\Prob{}{X_t^n = 1} \leq \delta$. Note that $X^n_t = X^n_{c\tau}$, $c$ being the largest integer such that $t \geq c \cdot \tau$. Since both the response rule and the schedule described above are memoryless, for every profile $\s$
$$
 \Prob{}{X^n_{c\tau} = 1 \mid X_{(c-1)\tau} = \s} = \Prob{}{X^n_{\tau} = 1 \mid X_0 = \s}.
$$
Let us use $\Prob{\s}{X_{\tau}^n = 1}$ as a shorthand for $\Prob{}{X^n_{\tau} = 1 \mid X_0 = \s}$. Moreover, let $\overline{B}$ denote the event that no bad move occurs in the interval $[1, \tau]$ and let $B_t$ denote the event that the first bad move occurs at time $t \in \{1,\ldots, \tau\}$. Then
$$
 \Prob{\s}{X_{\tau}^n = 1} = \Prob{\s}{X_{\tau}^n = 1 \mid \overline{B}} \Prob{}{\overline{B}} + \sum_{t=1}^\tau \Prob{\s}{X_{\tau}^n = 1 \mid B_t} \Prob{}{B_t}.
$$
Note that $B_t$ has probability at most $p$ and $\overline{B}$ has probability $(1-p)^{\tau}$. Obviously, $\Prob{}{X_{\tau}^n = 1 \mid \overline{B}} = 1$.
Moreover, by Observation~\ref{obs:adversary}, given a bad move of player $i_0 \neq n$ at time $t_{i_0}$, there is a sequence of time steps $t_{i_0} < t_{i_1}<t_{i_2}< \cdots < t_n$ such that player $i_j > i_{j-1}$ is selected at time $t_{i_j}$ and it is not selected further before $t_{i_{j+1}}$.
Therefore, player $i_1$ plays $0$ at time $t_{i_1}$ with probability $1-q$ because at that time $i_0$ is still playing $0$. Similarly, if player $i_j$ at time $t_{i_{j+1}}$ is still playing $0$, then player $i_{j+1}$ will play $0$ with probability $1-q$. Hence,
$$
 \Prob{}{X_{\tau}^n \neq 1 \mid B_t} \geq (1-q)^n.
$$
Then
$$
 \Prob{}{X_{\tau}^n = 1} \leq (1-p)^\tau + \tau p (1 - (1-q)^n) \leq \frac{1}{1+p\tau} + p\tau \frac{q}{1-q},
$$
where we repeatedly used that $1-x \leq e^{-x} \leq (1+x)^{-1}$.

The theorem follows by taking $p = \frac{1 - \delta}{\delta \cdot 2^{n-1}}$ and $q$ sufficiently small.
\end{proof}

\begin{remark}[BGP games]
\label{remark:bgp}
 The proof of the above theorem can be adapted to the class of BGP games~\cite{LSZ11,NSVZ11}. In these games, we are given an undirected graph with $n$ nodes (the players) and an additional destination  node $d$. The strategy of each player is to choose the \emph{next hop} in the routing towards $d$, that is, to point to one of its neighbors in the graph. Each node has a preference order over all possible paths (and ``non-valid paths'' have naturally the lowest rank).
 
 We consider the following instance with $n + 2$ players and the destination node $d$:
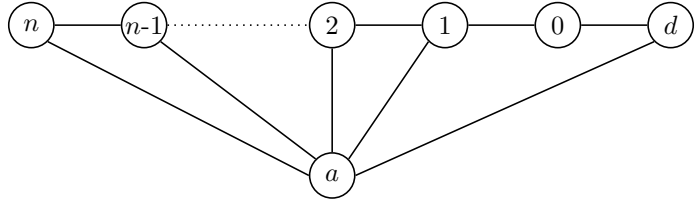
\begin{center}
 \begin{tikzpicture}[-,auto,on grid=true,semithick,
                     prof/.style={shape=circle,draw,inner sep=0pt,minimum size=6mm},
                     every label/.style={font=\scriptsize}]
  \node[prof] (A) {$n$};
  \node[prof] (B) [right=1.5cm of A] {$n$-$1$};
  \node[prof] (C) [right=2.5cm of B] {$2$};
  \node[prof] (D) [right=1.5cm of C] {$1$};
  \node[prof] (E) [right=1.5cm of D] {$0$};
  \node[prof] (F) [right=1.5cm of E] {$d$};
  \node[prof] (G) [below=2cm of C] {$a$};
  
  \draw (A.east) -- (B.west);
  \draw [dotted] (B.east) -- (C.west);
  \draw (C.east) -- (D.west);
  \draw (D.east) -- (E.west);
  \draw (E.east) -- (F.west);
  
  \draw (A.south east) -- (G.west);
  \draw (B.south east) -- (G.north west);
  \draw (C.south) -- (G.north);
  \draw (D.south west) -- (G.north east);
  \draw (F.south west) -- (G.east);
\end{tikzpicture}
\end{center}
Moreover, we specify the relationships among players according to the Gao-Rexford  model \cite{GR01}, as follows:
each node $i$, with $i = 1, \ldots, n-1$ is a provider for node $i+1$ and a customer of node $i - 1$;
moreover, node $a$ is a provider for each node $i = 1, \ldots, n$.

The so-called filtering policy of the Gao-Rexford model (that dictates that each node forwards traffic from one of its customers or to one of its customers only) implies that in this instance the only available path for nodes $0$ and $a$ is the one that routes directly to $d$. As for the player $i = 1,2,\ldots, n$, they are left with two available choice: we denote by $0$ the strategy to route to node $a$ and by $1$ the strategy to route to node $i - 1$. We can then set the preferences of player $i$ over the paths as follows:
\begin{itemize}
 \item Node $i$'s top ranked path is $i\rightarrow i-1\rightarrow \cdots \rightarrow 1 \rightarrow 0 \rightarrow d$.
 \item Node $i$'s second top ranked path is $i \rightarrow a \rightarrow d$.
 \item Every other path has lower rank.
\end{itemize}
The reader can check that in this instance we have that:
\begin{enumerate}
 \item For any strategy profile $\s$, $$u_i(\1_{1 \cdots i-1},0,\s_{i+1 \cdots n}) < u_i(\1_{1 \cdots i-1},1,\s_{i+1 \cdots n}),$$  
 by definition of $i$'s top ranked path above.
 \item For any $\s$ such that $\s_{1 \cdots i-1} \neq \1_{1 \cdots i-1}$ $$u_i(\s_{-i},1) < u_i(\s_{-i},0),$$
 by definition of $i$'s second top ranked path above.
\end{enumerate}
\end{remark}

\begin{remark}[logit dynamics]
 It is interesting that we can instantiate the abstract game described in the proof with utilities:
\begin{equation}
 u_i(0,\s_{-i}) = 0 \ \ \ \mbox{ and } \ \ \ u_i(1,\s_{-i})  = \begin{cases}
                     1, & \mbox{if $s_1=\cdots=s_{i-1}=1$}; \\
		    -L, & \mbox{otherwise;}
                    \end{cases}
\end{equation}
%
where $L$ is a large number. Similarly, the response rule described above may be instantiated as a logit response rule with noise $\beta$, that corresponds to set
\[p=\frac{1}{1+e^{\beta}} \ \ \ \mbox{ and } \ \ \
 q = \frac{e^{-\beta L}}{e^{\beta} + e^{-\beta L}}.
\]
Therefore our lower bound also applies to logit dynamics (for a suitable game and a suitable noise parameter $\beta$).
\end{remark}

The proof of Theorem~\ref{th:negative} highlights that imperfect dynamics differ from the perfect ones, since convergence result should necessarily depend on the schedule of players. Specifically, a closer look at the proof of Theorem~\ref{th:negative} shows that non-convergence of the imperfect best-response dynamics requires setting $p \approx \frac{1}{R}$ or greater. As a consequence, it may be possible to prove convergence to the equilibrium only by taking $p$ being smaller than $1/R$.

\subsection{A positive result (convergence time)}
Given the negative result above, we wonder whether there are values of $p$ for which the convergence of perfect best-response mechanisms is restored. The following theorem states that this occurs when $p$ is small with respect to parameters $R, \eta$ and $\ell$.
\begin{theorem}\label{th:convergence_solvable}
 For any NBR-solvable game $G$ and any small $\delta > 0$ an imperfect best-response dynamics converges to the Nash equilibrium of $G$ in $O(R \cdot \ell \log \ell)$ steps with probability at least $1 - \delta$, whenever $p \leq \frac{c}{\eta R \cdot \ell \log \ell}$, for a suitably chosen constant $c = c(\delta)$.
\end{theorem}
The following two lemmas represent the main tools in the proof of the theorem. Both lemmas hold for NBR-solvable games as for the more general class of NBR-reducible games. Moreover, in both lemmas we denote with $X_t$ the random variable that represents the profile of the game after $t$ steps of an imperfect best-response dynamics. Note also that, for an event $E$ we denote with $\Prob{\s}{E}$ the probability of the event $E$ conditioned on the initial profile being $X_0 = \s$, i.e., $\Prob{\s}{E}=\Prob{}{E\mid X_0=\s}$. 

\begin{lemma}\label{le:progress-probability}
For any profile $\s$, we have
 \begin{align}
 \Prob{\s}{X_{t+s} \in G_k \mid X_s \in G_k} & \geq 1 - \eta p t,
\label{eq:chain:stay-pne}
 \\
\label{eq:progress-probability}
 \Prob{\s}{X_{R+s} \in G_{k+1} \mid X_s \in G_k} & \geq 1 -\eta p R  - \epsilon.
\end{align}
\end{lemma}
\begin{proof}
Let the dynamics be in $G_k$ at time $s$ and observe that if the dynamics is not in $G_{k}$ at time $t+s$, then in one of steps in the interval $[s+1,s+t]$ some selected player played a NBR. Since at every step at most $\eta$ players are selected, \eqref{eq:chain:stay-pne} follows from the union bound.

Similarly, if the dynamics is not in $G_{k+1}$ at time $t+s$ given that player $i^{(k)}$ has been selected for update at least once during the interval $[s+1,s+t]$, then in one of these time steps some selected player played a NBR. Hence,
 \[
\Prob{\s}{X_{t+s} \not \in G_{k+1} \mid X_s \in G_k \cap SEL_{i^{(k)},s,t}} \leq \eta t p.
\]
Now simply observe that the definition of conditional probabilities and inequality above imply the following two inequalities:
\begin{align*}
\Prob{\s}{X_{t+s} \not \in G_{k+1} \mid X_s \in G_k} & \leq \Prob{\s}{X_{t+s} \not \in G_{k+1} \mid X_s \in G_k \cap SEL_{i^{(k)},s,t}}  + 1 - \Pr(SEL_{i^{(k)},s,t})\\
 & \leq \eta t p + 1 - \Pr(SEL_{i^{(k)},s,t}).
\end{align*}
Since $\Pr(SEL_{i^{(k)},s,R}) \geq 1 - \varepsilon$ by definition of imperfect best-response dynamics, the lemma follows.
\end{proof}

\begin{lemma}\label{le:hit:general}
For any profile $\s$  and $1 \leq k \leq \ell$, we have
$$
\Prob{\s}{X_{kR} \in G_k} \geq 1 - k \cdot (\eta p R + \varepsilon).
$$
\end{lemma}
\begin{proof}
 Observe that
 \begin{align*}
  \Prob{\s}{X_{kR} \notin G_k} & \leq \Prob{\s}{X_{kR} \notin G_{k} \mid X_{(k-1)R} \in G_{(k-1)R}} + \Prob{\s}{X_{(k-1)R} \notin G_{(k-1)R}}\\
  & \leq \eta pR + \varepsilon + \Prob{\s}{X_{(k-1)R} \notin G_{(k-1)R}},
 \end{align*}
where the first inequality follows from to the definition of conditional probabilities and the last one uses \eqref{eq:progress-probability}. Since $\Prob{\s}{X_{0} \notin G_{0}} = 0$ the lemma follows by iterating the argument.
\end{proof}

\begin{proof}[Proof (of Theorem~\ref{th:convergence_solvable})]
 Consider an interval $T$ of length $R \cdot \left\lceil \frac{\log (2\ell/\delta)}{\log(1/\varepsilon)}\right\rceil$. As discussed above, the probability that all players are selected at least once in an interval of length $T$ is $\frac{\delta}{2\ell}$. That is, every $(R,\epsilon)$-fair adversary is also $(T,\delta/2\ell)$-fair. The theorem thus follows from Lemma~\ref{le:hit:general} with $k=\ell$, $R=T$, $\varepsilon=\delta/2\ell$, and $p \leq \frac{\delta}{2} \cdot \frac{1}{\eta T \ell}$.
\end{proof}

\section{Incentive compatibility property}
In this section we ask if the incentive compatibility property holds also in presence of noise, that is, if deviating from a $p$-imperfect best-response rule is not beneficial for the player. Note that adopting a $p'$-imperfect response rule, with $p' < p$, should be not considered a deviation, since this rule is also a $p$-imperfect response rule.

\subsection{A negative result}
In this section we show that the incentive compatibility condition is not robust against noise and mistakes. Indeed, the following theorem says that,  even for arbitrarily small $p>0$, there are games for which $p$-imperfect best response dynamics are no longer incentive compatible (though best-response dynamics are incentive compatible according to the main result of \cite{NSVZ11} -- see Theorem~\ref{th:main:Nisan-et-al} above).

\begin{theorem}
\label{th:ic-negative}
 For any $p>0$, there exists a NBR-solvable game with clear outcome and a $p$-imperfect best-response dynamics whose response rule is not incentive compatible.
\end{theorem}
\begin{proof}
Consider the following game $G$ with clear outcome (the gray profile)

\begin{center}
\begin{game}{2}{2}
      & left         & right \\
    top  & \pne{$L+2,1$}    & $1,0$ \\
    bottom  & $0,0$    & $0,L$

\end{game}\hspace*{\fill}%
\end{center}
where $L$ is a sufficiently large constant (to be specified below). We examine what happens when both players play according to a particular $p$-imperfect best response dynamics, namely the logit dynamics with parameter $\beta$ (see Appendix~\ref{apx:logit} for details).
It turns out that the column player can improve her utility by playing \emph{always} strategy ``right'' (thus deviating to a $1$-imperfect best response).

We first use a known result by \cite{blumeGEB93} to compute the expected utilities of logit dynamics for this game. The above game is a \emph{potential game} and the potential function\footnote{A game is a \emph{potential game} if there exists a function $\Phi$
such that, for all $i$, for all  $\s_{-i}$, and for all $s_i,s_i'$, it holds that
$
 \Phi(s_i,\s_{-i}) - \Phi(s_i',\s_{-i})  = u_i(s_i',\s_{-i}) - u_i(s_i,\s_{-i}).
$ Such function $\Phi$ is called a \emph{potential function} of the game.} $\Phi$ is given by the following table:

\begin{center}
\begin{game}{2}{2}
      & left         & right \\
    top  & $L+2$    & $L+1$ \\
    bottom  & $0$    & $L$

\end{game}\hspace*{\fill}%
\end{center}
For potential games, the logit dynamics converges to the following distribution $\pi$ on the set of profiles:
\[
\pi(\s)= \frac{e^{-\beta \Phi(\s)}}{\sum_{\s'\in S} e^{-\beta \Phi(\s')}}.
\]
Hence, the expected utility of the column player when she plays according to the logit response rule is

\begin{equation}
\label{eq:expected-payoff:best-response}
 \frac{1 \cdot e^{\beta (L+2)} + L \cdot e^{\beta L}}{1 + e^{\beta L} + e^{\beta(L+1)} + e^{\beta(L+2)}} < \frac{e^{2\beta} + L}{1 + e^\beta + e^{2\beta}}.
\end{equation}
If instead the column player  plays always strategy ``right'', then her expected payoff is determined by the logit dynamics on the corresponding subgame. The same potential argument above (note that $\Phi$ is also a potential for this subgame) says that the column player has expected payoff equal to

\begin{equation}
\label{eq:expected-payoff:non-BR}
 \frac{L \cdot e^{\beta L}}{e^{\beta L} + e^{\beta(L+1)}} = \frac{L}{1 + e^\beta}.
\end{equation}
Since the right-hand side of \eqref{eq:expected-payoff:best-response} is smaller than \eqref{eq:expected-payoff:non-BR} for  $L\geq 1 + e^\beta$, the theorem follows from the observation that, in games with two strategies and integer payoffs, logit dynamics are $p_\beta$-imperfect with $p_\beta= \frac{1}{1+e^{\beta}}$ (see Section~\ref{sec:defs} and Appendix~\ref{apx:logit} for details).
\end{proof}

\subsection{A positive result}
As done for convergence, we now investigate for sufficient conditions for incentive compatibility. We will assume that utilities are non-negative: note that there are a lot of response rules that are invariant with respect to the actual value of the utility function and thus, in these cases, this assumption is without loss of generality.
Recall that we denote as $i^{(k)}$ and $G_k$ the first occurrence of the player and the corresponding subgame in the elimination sequence given by the definition of game with clear outcome (Definition~\ref{def:clear-outcome}).

It turns out that we need a ``quantitative'' version of the definition of clear outcome, i.e., that whenever the player $i$ has to eliminate a NBR her utility in the Nash equilibrium is sufficiently larger than the utility of any other profile in the subgame she is actually playing.
Specifically, we have the following theorem.
\begin{theorem}\label{th:p-BR:incentive-compatible}
 For any NBR-solvable game $G$ with clear outcome and any small $\delta > 0$, playing according to a $p$-imperfect rule is incentive compatible for player $i=i^{(k)}$ as long as $p \leq \frac{c}{\eta R \cdot \ell \log \ell}$, for a suitable constant $c=c(\delta)$, the dynamics run for $\Omega\left(R \cdot \ell \log \ell\right)$ and
 $$
  u_i(NE) \geq \frac{1}{1 - 2\delta} \bigg(2 \delta \cdot\max (u_i,G) + \max\left(u_i,G^{(k)}\right)\bigg),
 $$
 where $u_i(NE)$ is the utility of $i$ in the Nash equilibrium, $\max(u_i,G^{(k)}) = \max_{\s \in G^{(k)}} u_i(\s)$ and $\max(u_i,G) = \max_{\s \in G} u_i(\s)$.
\end{theorem}

We can summarize the intuition behind the proof of Theorem~\ref{th:p-BR:incentive-compatible} as follows:
\begin{itemize}
 \item If player $i$ always updates according to the $p$-imperfect response rule, then the game will be in the Nash equilibrium with high probability and hence her expected utility almost coincides with the Nash equilibrium utility;
 \item Suppose, instead, player $i$ does not update according to a $p$-imperfect response rule. Notice that elimination of strategies up to $G_k$ is not affected by what player $i$ does. Therefore profiles of $G \setminus G_k$ will be played only with small probability (but $i$ can gain the highest possible utility from these profiles), whereas the game will be in a profile of $G_k$ with the remaining probability.
\end{itemize}
Let us now formalize this idea. We start with the following lemma.
\begin{lemma}\label{lemma:prob_in_equilibrium}
 For any profile $\s$,  any $1 \leq k \leq \ell$ and  $t \geq kR$, we have
$$
\Prob{\s}{X_{t} \in G_k} \geq 1 - \eta p \cdot (t - l kR) - k \cdot (\eta p R + \varepsilon),
$$
where $l$ is the largest integer such that $t \geq l kR$.
\end{lemma}
\begin{proof}
We have
$$
 \Prob{\s}{X_{t} \notin G_k} \leq \Prob{\s}{X_{t} \notin G_k \mid X_{l kR} \in G_k} + \Prob{\s}{X_{l kR} \notin G_k}.
$$
From Lemma~\ref{le:progress-probability} we have
$$
 \Prob{\s}{X_{t} \notin G_k \mid X_{l kR} \in G_k} \leq \eta p \cdot (t - l kR).
$$
For $\s' = X_{(l-1)kR}$  Lemma~\ref{le:hit:general} implies
\[
 \Prob{\s}{X_{l kR} \notin G_k} = \Prob{\s'}{X_{kR} \notin G_k} \leq k \cdot (\eta p R + \varepsilon). \tag*{\qed}
\]
\let\qed\relax
\qed \end{proof}

\begin{remark}
 Observe that Lemma~\ref{lemma:prob_in_equilibrium} holds even if only players $i^{(1)}, \ldots, i^{(k)}$ are updating according to a $p$-imperfect response rule.
\end{remark}

By taking $T$ and $p$ as in the proof of Theorem~\ref{th:convergence_solvable} and by applying Lemma~\ref{lemma:prob_in_equilibrium} with $k=\ell$ and $(R,\varepsilon) = (T,\delta/2\ell)$ we have the following corollary.
\begin{cor}
\label{cor:nash}
 For any $\delta > 0$ and for any starting profile $\s$, if $t = \Omega\left(R \cdot \ell \log \ell\right)$, then
$$
 \Prob{\s}{X_{t} \in \hat G} \geq 1 - 2\delta.
$$
\end{cor}

\begin{proof}[Proof (of Theorem~\ref{th:p-BR:incentive-compatible})]
From Corollary~\ref{cor:nash}, the expected utility of $i$, given that all players are playing according to the $p$-imperfect response rule, will be at least $(1 - 2\delta) \cdot u_i(NE)$.

Suppose now that $i$ does not play a $p$-imperfect response rule. Similarly as done above, we let $T = R \cdot \left\lceil \frac{\log (2k/\delta)}{\log(1/\varepsilon)}\right\rceil$ and then, by applying Lemma~\ref{lemma:prob_in_equilibrium} with $(R,\varepsilon) = (T,\delta/2k)$ we obtain
$$
 \Prob{\s}{X_{t} \notin G_k} \leq 2 \delta.
$$
Hence, the expected utility of $i$ will be at most $2\delta\max(u_i,G) + \max(u_i,G^{(k)})$ and the theorem follows.
\end{proof}

\section{NBR-reducible Games}
For general games it is not possible to prove convergence to a pure Nash equilibrium without making additional assumptions on the schedule and on the response rule. Such negative result applies also to NBR-reducible games, the natural extension of NBR-solvable ones. 
However, we shall see below that, for NBR-reducible games, several questions on the dynamics of a game $G$ can be answered by studying the dynamics of the reduced game $\hat G$.

\subsection{Impossibility of general results}\label{subsec:impossible}
At first sight NBR-solvable games may appear as a limited class of games. The class of NBR-reducible games is a natural generalization that can be applied to more settings.  Unfortunately, we next show that no general result can be stated about the convergence of $p$-imperfect best-response dynamics. Specifically, we will show that the system behaves differently not only with respect to which schedule is adopted, but also with respect to how a $p$-imperfect response rule is implemented.

\paragraph{Same response rule and different schedules.}
Consider the classical  $2$-player coordination game:
$$
\begin{game}{3}{2}
      & 0     & 1\\
0   &1,1    &0,0\\
1   &0,0    &1,1
\end{game}
$$
Obviously, the game above is not NBR-solvable but it may be the reduced game for a NBR-reducible game. We assume players update according to the best-response response rule, but we consider two different schedules: the first one select just one player randomly at each time step, whereas the second one update all players at same time. Then, it is easy to see that the dynamics with the first schedule always converges in one of the two Nash equilibria of the game, $(0,0)$ and $(1,1)$, whereas with the second one it can cycle over profiles $(1,0)$ and $(0,1)$ and never converges (see e.g. Alos-Ferrer and Netzer~\cite{afnGEB10}).

\paragraph{Same schedule and different response rules.}
Consider the following $2$-player game:
$$
\begin{game}{3}{2}
      & 0     & 1\\
0   &0,0    &0,1\\
1   &0,1    &1,0
\end{game}
$$
The game has an unique Nash equilibrium, namely $(1,0)$, but it is not NBR-solvable. We consider two different response rules: in both players update according to the best-response rule, but they differ in how handling multiple best-response strategies. The first response rule states that one of these strategies is selected randomly (as happen in logit response rule), whereas the second response rule choose the best response randomly only if the current strategy is not a best response. In these cases, the second response rule adopts a conservative approach and chooses the current strategy (this is exactly the way the mutation model handle multiple best responses). Let us pair these response rules with the schedule that select just one player randomly at each time step. Then, it is easy to see that the dynamics with the second response rule always converges to the Nash equilibrium and never move from there, whereas with the first one it cycles infinitely over profiles $(1,0), (0,0), (0,1)$ and $(1,1)$ (see e.
g. Alos-Ferrer and Netzer~\cite{afnGEB10}).

\subsection{Reductions between games}\label{sec:reductions}
We will see now that some of the ideas developed in the previous sections about NBR-solvable games and their pure Nash equilibria can be extended to address questions about NBR-reducible games. In particular, we will see that for a wide class of questions about imperfect best-response dynamics for a NBR-reducible game $G$, an answer can be given simply by considering a restriction of this dynamics to the reduced game $\hat G$.

Before formally stating this fact, let us introduce some useful concepts.

\paragraph{The dynamics as a Markov chain.}
Recall that, in general, the dynamics is not a Markovian process (for instance, the schedule or the response rule of the players could depend on the history -- the past strategy profiles). 
In order to state our results in full generality, we allow both the schedule and the response rule to depend on what we call the \emph{status} of the dynamics, that is, on a set of information in addition to the  current strategy profile.

We say that the dynamics is in a \emph{status--profile pair} $(h,\s)$ if $h$ is the set of information currently available and $\s$ is the profile currently played. We denote with $H$ the set of all status--profile pairs $(h,\s)$ and with $\hat H$ only the ones with $\s \in \hat G$. Let $X_t$ be the random variable that represents the status--profile pair $(h,\s)$ of the dynamics after $t$ steps of the imperfect best-response dynamics. Then, for every $(h,\s), (h',\s') \in H$ we set
$$
 P\big((h,\s),(h',\s')\big) = \Pr\Big(X_1 = (h',\s') \mid X_0 = (h,\s)\Big).
$$
That is, $P$ is the transition matrix of a Markov chain on state space $H$ and it describes exactly the evolution of the dynamics. Note that we are not restricting the dynamics to be memoryless, since in the status we can save the history of all previous iterations.
For a set $A \subseteq H$ we also denote $P\big((h,\s),A\big) = \sum_{(h',\s') \in A} P\big((h,\s),(h',\s')\big)$.

\begin{remark}\label{rem:markovian:initial-observation}
Observe that, for any starting status-profile $(h,\s)\in H$, if we consider only the profile component of each $X_t$, this describes a particular imperfect best response dynamics (every $h$ gives a different dynamics) starting at $\s$.
Denoting by $X_t^{(h)}$ the random variable corresponding to this dynamics  (now $X_t^{(h)}$ is a strategy profile), we can import prior bounds via the following useful inequality:
\begin{equation}\label{eq:markovian:initial-observation}
  \Prob{}{X_t = (h',\s')\mid X_0=(h,\s)} \leq \sum_{(h'',\s')\in H}\Prob{}{X_t = (h'',\s')\mid X_0=(h,\s)} = 
 \Prob{\s}{X_t^{(h)} = \s'}
\end{equation}
Clearly, the response rule and the schedule of the imperfect best-response dynamics described by $X_t^{(h)}$
inherit the same parameters of the original dynamics $X_t$.
\end{remark}

\paragraph{The restricted dynamics.}
As mentioned above, we will compare the original dynamics with a specific restriction on the subset $\hat H$ of status--profile pairs. Now we describe how this restriction is obtained. Henceforth, whenever we  refer to the restricted dynamics, we  use $\hat X_t$ and $\hat P$ in place of $X_t$ and $P$.
Then, the restricted dynamics is described by a Markov chain on state space $H$ with transition matrix $\hat P$ such that for every $(h,\s), (h',\s') \in H$
$$
 \hat P\big((h,\s),(h',\s')\big) = \begin{cases}
                           \frac{P((h,\s),(h',\s'))}{P((h,\s),\hat H)}, & \mbox{if } (h,\s), (h',\s') \in \hat H;\\
                           0, & \mbox{otherwise.}
                          \end{cases}
$$
Thus, the restricted dynamics is exactly the same as the original one except that the first never leaves the subgame $\hat G$, whereas in the latter, at each time step, there is probability at most $p$ to leave this subgame.
The following lemma quantifies this similarity, by showing that, for every $(h,\s) \in \hat H$, the \emph{total variation distance} (TV)\footnote{See Appendix~\ref{apx:tvdistance} for a review of the main properties of the total variation distance.} between the original and the restricted dynamics starting from $(h,\s)$ is small.
\begin{lemma}\label{lemma:restricted}
For every $(h,\s) \in \hat H$,
\begin{equation}\label{eq:closeness}
 \tv{P^t\big((h,\s), \cdot\big) - \hat P^t\big((h,\s), \cdot\big)} \leq \eta p t.
\end{equation}
\end{lemma}
\begin{proof}
First observe that, since $(h,\s) \in \hat H$, it holds that 
\begin{equation}
 \label{eq:lemma:restricted:obs}
\sum_{(h',\s') \in H \setminus \hat H} P((h,\s), (h',\s')) \leq \sum_{(\hat h,\hat \s)\in \hat H} \sum_{(h',\s') \in H \setminus \hat H} P((\hat h,\hat \s), (h',\s')) \leq \Prob{\s}{X_1^{(h)}\not\in \hat G\mid X_0^{(h)}\in \hat G} \leq \eta p,
\end{equation}
where the last two inequalities follow from Remark~\ref{rem:markovian:initial-observation} and  Lemma~\ref{le:progress-probability}, respectively.
We next prove \eqref{eq:closeness} by induction on $t$. The base case is $t=1$ for which the set of status--profile pairs $(h',\s')$ such that $P((h,\s), (h',\s')) > \hat P((h,\s), (h',\s'))$ is exactly $H \setminus \hat H$ and thus
\begin{equation}
\label{eq:lemma:restricted:base-case}
\begin{aligned}
 \tv{P\big((h,\s), \cdot\big) - \hat P\big((h,\s), \cdot\big)} & = \sum_{(h',\s') \in H \setminus \hat H} \Big(P((h,\s), (h',\s')) - \hat P((h,\s), (h',\s'))\Big)\\
  & = \sum_{(h',\s') \in H \setminus \hat H} P((h,\s), (h',\s')) \leq \eta p,
\end{aligned}
\end{equation}
where the last inequality follows from \eqref{eq:lemma:restricted:obs}.
As for the inductive step, we have 
\begin{align*}
\tv{P^t\big((h,\s), \cdot\big) - \hat P^t\big((h,\s), \cdot\big)} & \leq \quad \mbox{(TV triangle inequality)} \\
 \tv{P\big((h,\s), \cdot\big)P^{t-1} - \hat P\big((h,\s), \cdot\big)P^{t-1}} + \, \tv{\hat P\big((h,\s), \cdot\big)P^{t-1} - \hat P\big((h,\s), \cdot\big)\hat P^{t-1}} & \leq \quad \mbox{(TV monotonicity)}\\
\tv{P\big((h,\s), \cdot\big) - \hat P\big((h,\s), \cdot\big)} + \, \sup_{(h',\s') \in \hat H} \tv{P^{t-1}\big((h',\s'), \cdot\big) - \hat P^{t-1}\big((h',\s'), \cdot\big)} & \leq \quad 
\mbox{(induction)} \\
\eta p+\eta p(t-1) = \eta  pt. \tag*{\qed} & 
\end{align*}
\let\qed\relax
\end{proof}

\paragraph{Status--profile events.}
We now describe the kind of questions about imperfect best-response mechanisms and NBR-reducible games for which a reduction can be beneficial. Roughly speaking, these are all questions about the occurrence (and the time needed for it) of events that can be described only by looking at status--profile pairs.

Specifically, a \emph{status--profile set event} for an imperfect best-response dynamics is a set of status--profile pairs. A \emph{status--profile distro event} for an imperfect best-response dynamics is a distribution on the status--profile pairs. More generally, we refer to \emph{status--profile event} if we do not care whether it is a set or a distro event.
Note that many equilibrium concepts can be described as status--profile events, like Nash equilibria, sink equilibria~\cite{gmvFOCS05}, correlated equilibria~\cite{aumann1974} or logit equilibria~\cite{afppSAGT10j}: in any case we should simply list the set of states or the distribution over states at which we are interested in. Properties like ``a profile that is visited for $k$ times'' or ``a cycle of length $k$ visited'' are other examples of status--profile events. We remark that in these  examples it is crucial that the equilibrium is defined on the status--profile pairs and not just on the profiles: indeed, the status can remember the history of the game and identify such events, whereas they are impossible to recognize if we only know the current profile.

For an NBR-reducible game $G$, a status--profile set event is \emph{reducible} if the set of status--profile pairs that represent the event contains some profile from $\hat G$. A status--profile distro event is \emph{reducible} if status--profile pairs on which is defined the distribution that represent the event contains only profiles of $\hat G$. It turns out that each one of the equilibria concepts described above is a reducible status-profile event: indeed, since all profiles not in $\hat G$ contain NBR strategies,  they are not in the support of any Nash, any sink and any correlated equilibrium; as for the logit equilibrium (that assigns non-zero probability to profiles not in $\hat G$) it is not difficult to show (see Appendix~\ref{apx:stationary_close}) that the logit equilibrium of $G$ is close to the logit equilibrium of $\hat G$.

A status--profile set event \emph{occurs} if the imperfect best-response dynamics reaches a status--profile pair in the set of pairs that represent the event.
Similarly, a status--profile distro event \emph{occurs} if the distribution on the set of profiles generated by the dynamics is close to the one that represent the event. The \emph{occurrence time} of a status--profile event is the first time step in which it occurs.

\bigskip
We are now in a position to state the main result on NBR-reducible games:
\begin{theorem}\label{th:reducible}
 For any NBR-reducible game $G$ and any small $\delta > 0$,
if a reducible status--profile event for an imperfect best-response dynamics occurs in the restricted dynamics, then it occurs with probability at least $1 - \delta$.
 Moreover, let us denote with $\tau$ the occurrence time of the event $E$ in the restricted dynamics. Then, $E$ occurs in the original dynamics in $O(R \cdot \ell \log \ell + \tau)$ steps with probability at least $1 - \delta$, whenever $p \leq \min \left\{\frac{c_1}{\eta R \cdot \ell \log \ell}, \frac{c_2}{\eta \tau}\right\}$, for suitable constants $c_1=c_1(\delta)$ and $c_2=c_2(\delta)$.
\end{theorem}
\begin{proof}
 We will show that the dynamics will be in $\hat H$ after $O(R \cdot \ell \log \ell)$ with probability at least $1-\delta/2$; moreover, if the dynamics is in $\hat H$ after  $t$ steps, then a reducible status--profile event occurs in further $\tau$ steps with probability at least $1-\delta/2$. Hence, the probability that the event does not occurs in $O(R \cdot \ell \log \ell + \tau)$ steps will be at most $\delta$ and the theorem follows.

 Specifically, consider an interval $T$ of length $R \cdot \left\lceil \frac{\log (4\ell/\delta)}{\log(1/\varepsilon)}\right\rceil$. 
 From Remark~\ref{rem:markovian:initial-observation} and by applying Lemma~\ref{le:hit:general} with $k=\ell$, $(R,\varepsilon) = (T,\delta/4\ell)$ and $p \leq \frac{\delta}{4} \cdot \frac{1}{\eta T \ell}$ we have that for every $(h,\s) \in H$
 $$
  \Prob{}{X_{\ell T} \not\in \hat H \mid X_0 = (h,\s)} \leq \Prob{\s}{X_{\ell T}^{(h)} \not \in \hat G} \leq \delta/2.
 $$
 Finally, note that the probability that, for every $t > 0$, a reducible status--profile event occurs in $t+\tau$ steps given that after $t$ steps it is in $(h',\s') \in \hat H$, is the same as if we assume the dynamics starts in $(h',\s')$, i.e., it is equivalent to the probability that the event occurs in $\tau$ steps from $(h',\s')$.
 If the event $E$ is a distro event, i.e. the restricted dynamics converges after $\tau$ steps to a distribution $\pi$ on the status--profile pairs, then, from~\eqref{eq:closeness}, the distribution after $\tau$ steps of the original dynamics is $\pi$ except for an amount of probability of at most $\eta p \tau$. On the other hand, if the event $E$ is a set event, i.e. the restricted dynamics converges after $\tau$ steps to a set $A$ of status--profile pairs, then, from~\eqref{eq:closeness} we have
 $$
  \Pr\Big(X_\tau \in A\Big) \geq \Prob{}{\hat X_\tau \in A} - \mu p \tau = 1 - \mu p \tau,
 $$
 and hence, after $\tau$ steps, the original dynamics is in $A$ except with probability at most $\mu p \tau$.
 Then, by Lemma~\ref{lemma:restricted} and by taking $p \leq \frac{\delta}{2} \cdot \frac{1}{\eta \tau}$, the probability that the event occurs in $\tau$ steps for the original dynamics starting from $(h',\s')$ is at least $1 - \delta/2$.
\end{proof}

\begin{example}[interference games]
 The following is an instance of \emph{interference games} \cite{AulMosPenPer08} which model interference in wireless networks according to the Signal-to-Interference-plus-Noise-Ratio (SINR) physical model.  Formally, the game has the following utilities:
 \[
  u_i(\s) = \begin{cases}
                           v_i - s_i  & \mbox{if } s_i > \sum_{j\neq i} s_j\\
                           - s_i & \mbox{otherwise}
                 \end{cases},
 \]
where $v_i$ is some value assigned to each player (representing the utility if the transmission succeeds), strategy $s_i$ is the energy spent 
for transmitting, and the condition above indicates whether $i$'s transmission is successful received or not. For two players with $v_1=v_2=1+\gamma$, $0<\gamma<1$, and $S_i=\{0,1,2,3\}$ the payoffs are
$$
\begin{game}{4}{4}
      & 0     & 1 & 2 & 3 \\
0    & 0,0    & 0,$\gamma$ & 0, $\gamma - 1$ & 0, $\gamma - 2$ \\
1    & $\gamma$,0    & -1,-1 & -1,$\gamma - 1$ & -1,$\gamma - 2$ \\
2    & $\gamma - 1$,0    & $\gamma - 1$,-1 & -2,-2 & -2,$\gamma - 2$ \\
3    & $\gamma - 2$,0    & $\gamma - 2$,-1 & $\gamma - 2$,-2 & -3,-3
\end{game}
$$
which is NBR-reducible to this subgame:
$$
\begin{game}{2}{2}
      & 0     & 1  \\
0    & 0,0    & 0,$\gamma$  \\
1    & $\gamma$,0    & -1,-1  
\end{game}
$$
Note that this is a potential game, with potential $\Phi$ given by the following table:
$$
\begin{game}{2}{2}
      & 0     & 1  \\
0    & 0    & $-\gamma$  \\
1    & $-\gamma$    & $1-\gamma$  
\end{game}
$$
The logit dynamics with parameter $\beta$ on this subgame converges to the following distribution $\hat \pi$:
$$
\begin{game}{2}{2}
      & 0     & 1  \\
0    & $1/Z$    & $e^{\beta\gamma}/Z$  \\
1    & $e^{\beta\gamma}/Z$    & $e^{\beta(\gamma-1)}/Z$  
\end{game}
$$
where $Z = 1+ 2e^{\beta\gamma} +e^{\beta(\gamma -1)}$. This means that, for growing values of $\beta$, the dynamics ``converges'' to the two Nash equilibria of the game (the stationary distribution of profiles $(0,1)$ and $(1,0)$ increase whereas the probability of the other profiles diminishes). 
 
According to Theorem~\ref{th:reducible},  the logit dynamics with parameter $\beta$ on the original interference game can be analyzed by looking at the dynamics on the subgame only, provided $p$ being small enough. (Note that in general both $p$ and $\tau$ in Theorem~\ref{th:reducible} depend of $\beta$.) For instance, taking the event $E=$``the dynamics is in one of the two PNE'', $\tau$ is the hitting time of such PNE's and, for $\beta$ large enough, the original dynamics converges to a PNE in $O(R+\tau)$ time steps.  
\end{example}

The reader may have noticed that, in the previous example, the parameters are set so as to obtain a subgame which is a potential game. Indeed, the class of interference games \cite{AulMosPenPer08} does not posses pure Nash equilibria in general. For instance, setting $v_1=v_2=3$ and $S_i=\{0,1,\ldots,4\}$ would give an interference game which is NBR-reducible to a subgame containing a sink equilibrium \cite{gmvFOCS05} (but no pure Nash equilibrium). Though the subgame is not a potential game, Theorem~\ref{th:reducible} suggests a natural way to tackle the problem by analyzing the sink equilibria of the subgame (e.g, the time for the restricted dynamics to reach such an equilibrium).

\medskip
We conclude this section by observing that Theorem~\ref{th:reducible} is similar in spirit to the results by Candogan et al.~\cite{NearPotG}.
Both results aims to link the behavior of dynamics for a game $G$ to the behavior of the dynamics for a close game $\hat{G}$,
for which the dynamics is well-understood.
In \cite{NearPotG} are considered games that are ``close'' to potential games and
it has been showed that the behavior of best and better response dynamics, logit dynamics and fictitious play
can be approximated by the equilibrium behavior of these dynamics on the close potential games.
The extent at which this approximation is good depends on the distance between the games and
on the number of profiles of these games (that is exponential on the number of players).
Our result instead refers to generic imperfect response dynamics
(and thus they holds for best response and logit dynamics but not for better response dynamics and fictitious play),
but we consider only a very specific ``closeness measure'', namely that the game $G$ is NBR-reducible to the game $\hat{G}$.
On the other side, the game $\hat{G}$ can approximate $G$ arbitrarily well,
not only with respect to the equilibrium behavior, but also with respect to the transient behavior of the dynamics.

\section{Conclusions}
This work addresses how a (small) probability of selecting a non-best response by the players can affect the convergence and incentive compatibility properties of a class of dynamics and mechanisms in \cite{NSVZ11}. We have first show that there are games for which convergence occurs only for 
a probability $p$ \emph{exponentially} small in the number of players (Theorem~\ref{th:negative}). Our positive result (Theorem~\ref{th:convergence_solvable}) says that convergence can be achieved if $p$ is ``sufficiently small'' compared to certain parameters of the game and of the schedule (the number $\ell$ of rounds needed in the original dynamics without mistakes and the length $R$ of a ``probabilistic round'' of the schedule -- see Definition~\ref{def:fair-selection}). Note that, in this case, the convergence time 
of the imperfect best-response dynamics is only moderately larger than the convergence time of the perfect best-response dynamics: When the latter is guaranteed to converge in $R\ell$ time steps (Theorem~\ref{th:main:Nisan-et-al}) the former converges in $O(R\ell \log \ell)$ time steps with good probability (Theorem~\ref{th:convergence_solvable}). 

As for  incentive compatibility, we showed that imperfect best-response dynamics require a stronger condition than the one sufficient for perfect best-response dynamics (cf. Theorems~\ref{th:main:Nisan-et-al} and \ref{th:ic-negative}). The sufficient condition (Theorem~\ref{th:p-BR:incentive-compatible}) is essentially a quantitative version of the condition in \cite{NSVZ11} (clear outcome) which requires the payoffs at the equilibrium to be ``sufficiently'' larger than at other profiles.

Finally, we suggest a natural extension to games which are not NBR-solvable, and the elimination of NBR strategies yields only a subgame. 
In such cases, it is natural to consider the dynamics of the subgame as a good approximation of the dynamics of the original game. 
Theorem~\ref{th:reducible} gives a quantitative bound on the time to reach a certain ``event'' as the sum of two quantities: The time for the dynamics to converge to the subgame plus the time the dynamics restricted to the subgame takes to reach the same event.  
We feel Theorem~\ref{th:reducible} might have several applications to games which are reducible to a subgame whose structure makes the analysis of the dynamics easier. For instance, if the subgame is a potential game, logit dynamics of the original game (even if this is not a potential game) can be accurately described by the logit dynamics of the subgame which is known to have an explicit simple stationary distribution \cite{blumeGEB93}.\footnote{In the conference version of this work \cite{imperfectSAGT}, we claimed that in a modified version of PageRank games \cite{hs2008},
there exists a subgame which is a potential game and thus our results can be combined with \cite{blumeGEB93}
to obtain a good approximation of the logit dynamics for these games.
Unfortunately, this claim was wrong and the logit dynamics for this subgame is in general not easy to analyze.} Also, when the subgame $\hat G$ is \emph{not} a potential game, it might be possible to combine our results with those in \cite{NearPotG} in a two-steps analysis: Reduce the game $G$ to a subgame $\hat G$  and then study this game by considering the ``closest'' potential game $\tilde G$. More in general, the advantage of reducing the analysis to a subgame $\hat G$, is that the resulting dynamics can be simpler to analyze using known techniques. For instance, for some potential games,  it is possible to analyze variants of the logit dynamics that differ in the schedule of the players \cite{afnGEB10,AulFerPasPenPer11}.

\section*{Acknowledgments}
Part of the work has been done when both authors were at the Universit\`a di Salerno and later when the first author was at Universit\'e Paris Dauphine.

\bibliographystyle{plain}
\bibliography{imperfect-best-response}

\newpage
\appendix

\section{Models for limited knowledge and bounded rationality}
\subsection{Mutation and mistakes models}\label{apx:mistakes}
The \emph{mutation} and the  \emph{mistakes} model adopt the same response rule: at every time step, each selected player updates her strategies to the best response to what other players are currently doing except with probability $\varepsilon$. With such a probability, a \emph{mutation} or \emph{mistake} occurs, meaning that the selected player choose a strategy uniformly at random. That is, suppose player $i$ is selected at time step $t$ and the current strategy profile is $\s^t$. We denote with $b_i(\s^t)$ the best response of $i$ to profile $\s^t$ (if more than one best response exists and the current strategy $x^t_i$ of $i$ is a best response, then we set $b_i(\s^t) = x^t_i$, otherwise we choose one of the best response uniformly at random). Then, a strategy $s_j \in S_i$ will be selected by $i$ with probability
$$
 p_{ij} = \begin{cases}
           (1-\varepsilon) + \varepsilon \cdot \frac{1}{|S_i|}, & \text{if } s_j = b_i(\s^t);\\
           \varepsilon \cdot \frac{1}{|S_i|}, & \text{otherwise.}
          \end{cases}
$$
The main difference between these models concerns the schedule: the mutation model assumes that at each time step every player is selected for update; the mistakes model assumes that at each time step only one player is selected uniformly at random for update.

\subsection{Logit dynamics}\label{apx:logit}
Logit dynamics is another kind of imperfect best response dynamics in which the probability of deviating from best response is determined by a parameter $\beta\geq 0$. 
The \emph{logit dynamics} for a game $G$ with parameter $\beta$ runs as follows. At every time step 
\begin{enumerate}
 \item Select one player $i$ uniformly at random;
 \item Update the strategy of player $i$ according to the following probability distribution:
 \[
 \sigma_i(s_i,\s_{-i}) = \frac{e^{\beta u_i(s_i,\s_{-i})}}{Z(\s_{-i})},
\]
where $Z(\s_{-i})=\sum_{s_i'}e^{\beta u_i(s_i',\s_{-i})}$ and $\beta\geq 0$. 
\end{enumerate}

\begin{remark}\label{re:logit-imperfect}
 The parameter $\beta$ is sometimes called the ``inverse noise'' parameter: for $\beta=0$ the player chooses the strategy uniformly at random, while for $\beta\rightarrow \infty$ the player chooses a best response (if more than one best response is available, she selects one of these uniformly at random). Thus, the probability that $i$ plays a best response is guaranteed to be least $1/|S_i|$. So,  every logit dynamics is 
 $p$-imperfect for some $p<1$. Moreover, if the payoffs between a best response and a non-best response differ by at least $\gamma$, then the dynamics is $p$-imperfect with 
 $
 p \leq ({m-1})/({m-1 + e^{\gamma\beta}}).
 $
\end{remark}

The above dynamics defines an ergodic Markov chain \cite{blumeGEB93} over the set of all strategy profiles and with transition probabilities $P$ given by 
\begin{equation}\label{eq:transmatrix}
P(\s, \s') = \frac{1}{n} \cdot
\begin{cases}
\sigma_i(s_i ,\s_{-i}), & \quad \mbox{ if } \s_{-i} = \s'_{-i} \mbox{ and } s_i \neq s'_i; \\
\sum_{i=1}^n \sigma_i(s_i , \s_{-i}), & \quad \mbox{ if } \s = \s'; \\
0, & \quad \mbox{ otherwise.}
\end{cases}
\end{equation}
Since the Markov chain is ergodic, the chain converges to its (unique) \emph{stationary distribution} $\pi$. That is, for any starting profile $\s$ and any profile $\s'$,  
\[
 \lim_{t \rightarrow \infty} P^t(\s,\s') = \pi(\s'),
\]
where $P^t(\s,\s')$ is the probability that the dynamics starting with profile $\s$ is in profile $\s'$ after $t$ steps. Thus, the stationary distribution represents the equilibrium reached by the logit dynamics, also called 
the \emph{logit equilibrium} of game $G$ \cite{afppSAGT10j}.

The stationary distribution of logit dynamics is fully understood for the class of so-called \emph{potential games}.
We recall that a game is a potential game if there exists a function $\Phi$ such that, for all $i$, for all  $\s_{-i}$, and for all $s_i,s_i'$, it holds that
\[
 \Phi(s_i,\s_{-i}) - \Phi(s_i',\s_{-i})  = u_i(s_i',\s_{-i}) - u_i(s_i,\s_{-i}).
\]
Blume~\cite{blumeGEB93} showed that in this case, the stationary distribution of the corresponding logit dynamics is 
$$
\pi(\s) = \frac{e^{-\beta \Phi(\s)}}{Z} ,
$$
where
$Z = \sum_{\s' \in S} e^{-\beta \Phi(\s')}$
is the normalizing constant.

Finally, the time to converge to the logit equilibrium is 
the so-called \emph{mixing time} of the corresponding Markov chain \cite{afppSAGT10j}, defined as follows. For any $\epsilon>0$, consider 
\begin{equation}
 \label{eq:t_mix} 
 t_{mix}(\epsilon) := \min_{t \in \mathbb{N}}\max_{\s\in \Omega} \left\{\tv{P^t(\s,\cdot) - \pi} \leq \epsilon\right\},
\end{equation}
where $\tv{P^t(\s,\cdot) - \pi} = \frac{1}{2}\sum_{\s'' \in \Omega} |P^t(\s,\s'') - \pi(\s'')|$ is the total variation distance (see Section~\ref{apx:tvdistance} for more details). This quantity measures the time needed for the chain (dynamics) to get close to stationary within an additive factor of $\epsilon$. The \emph{mixing time} of the chain is simply defined as $t_{mix}:=t_{mix}(1/2)$, since it is well-known that  $t_{mix}(\epsilon) \leq t_{mix}\cdot \log (1/\epsilon)$ for any $\epsilon$.

\section{Total variation distance}\label{apx:tvdistance}
The \emph{total variation distance} between distributions $\mu$ and $\hat \mu$ on an enumerable state space $\Omega$ is
\[
 \tv{\mu - \hat \mu}:= \frac{1}{2}\sum_{x \in \Omega} \mid \mu(x) - \hat \mu(x)| = \sum_{\begin{subarray}{c}x \in \Omega\\\mu(x) > \hat \mu(x)\end{subarray}} \mu(x) - \hat \mu(x).
\]
Note that the total variation distance satisfies the usual triangle inequality of distance measures, i.e.,
\[
 \tv{\mu - \hat \mu} \leq \tv{\mu - \mu'} + \tv{\mu' - \hat \mu},
\]
for any distributions $\mu$ and $\mu'$.
Moreover, the following  monotonicity properties hold:
\begin{align}
 \tv{\mu P - \hat \mu P} & \leq \tv{\mu - \hat \mu}, \label{eq:samePin_mu}\\
 \tv{\mu P - \mu \hat P} & \leq \sup_{x \in \Omega} \tv{P(x,\cdot) - \hat P(x,\cdot)}, \label{eq:same_mu}\\
 \tv{\mu P - \hat \mu P} & \leq \sup_{x,y \in \Omega} \tv{P(x,\cdot) - P(y,\cdot)}, \label{eq:samePinP}
\end{align}
where $P$ and $\hat P$ are stochastic matrices.
Indeed, as for \eqref{eq:samePin_mu} we have
\begin{align*}
 \tv{\mu P - \hat \mu P} = \tv{(\mu - \hat \mu) P} & = \frac{1}{2} \sum_{y \in \Omega} \left|  \sum_{x \in \Omega} (\mu(x) - \hat \mu(x)) P(x,y)\right|\\
 & \leq \frac{1}{2} \sum_{x \in \Omega} \left| \mu(x) - \hat \mu(x)\right| \sum_{y \in \Omega} P(x,y)\\
 & = \tv{\mu - \hat \mu}.
\end{align*}
As for \eqref{eq:same_mu} we observe that
\begin{align*}
\tv{\mu P - \mu \hat P} = \tv{\mu (P - \hat P)} & = \frac{1}{2} \sum_{y \in \Omega} \left| \sum_{x \in \Omega} \mu(x) (P(x,y) - \hat P(x,y))\right|\\
 & \leq  \sum_{y \in \Omega}  \left( \frac{1}{2} \sum_{x \in \Omega} \mu(x)\left| P(x,y) - \hat P(x,y)\right| \right)\\
 \\
 & =  \sum_{x \in \Omega}  \mu(x)\left( \frac{1}{2} \sum_{y \in \Omega} \left| P(x,y) - \hat P(x,y)\right| \right)\\
 & \leq \sup_{x \in \Omega} \tv{P(x,\cdot) - \hat P(x,\cdot)}.
\end{align*}
Finally, for \eqref{eq:samePinP} we have
\begin{align*}
\tv{\mu P - \hat \mu P} & = \tv{\sum_{z \in \Omega} \mu(z) \sum_{w \in \Omega} \hat \mu(w) \left( P(z,\cdot) - P(w, \cdot) \right)}\\
& \leq \sum_{z \in \Omega} \mu(z) \sum_{w \in \Omega} \hat \mu(w)  \tv{P(z,\cdot) - P(w, \cdot)}\\
& \leq \sup_{x,y \in \Omega} \tv{P(x,\cdot) - P(y,\cdot)}.
\end{align*}

\section{Equilibria of logit dynamics and NBR-reducible games}\label{apx:stationary_close}
In this section we specialize the approach described in Section~\ref{sec:reductions} to the case of logit dynamics. 
First of all, the status ``$h$'' is immaterial in this case since the dynamics is by definition Markovian (i.e., the transition probabilities depend only on the current profile).

We next provide a sufficient condition for which the (equilibrium of) the logit dynamics of the subgame 
is a good approximation of the (equilibrium of) the  logit dynamics of the original game.
Let $G$ be a game NBR-reducible to $\hat G$, and let $\pi_\beta$ and $\hat \pi_\beta$ denote 
the stationary distributions of the logit dynamics with parameter $\beta$ for $G$ and $\hat G$, respectively. 
The following lemma says that $\pi_\beta$ and $\hat \pi_\beta$ are close to each other (in total variation) if $\beta$ is large enough.

\begin{lemma} Let $\hat \tau_\beta$ be the mixing time of the restricted chain given by the logit dynamics with parameter $\beta$, and let $p=p_\beta$ be the corresponding probability of selecting a NBR strategy. If $\lim_{\beta \rightarrow \infty} (p_\beta\hat\tau_\beta) = 0$, then for every $\delta > 0$, there exists a constant $\beta_\delta$ such that
 $$
  \tv{\pi_\beta - \hat \pi_\beta} \leq \delta
 $$
 for all $\beta\geq \beta_\delta$.
\end{lemma}
\begin{proof}
 Let $\tau = \hat t_{mix}^{(\beta)}(\delta/8)$ be the mixing time of the restricted chain. Consider first two copies of the chain starting in profiles $\hat \s, \hat \s' \in \hat G$ and bound the total variation after $\tau$ time steps:
\begin{align*}
 \tv{P^\tau(\hat \s,\cdot )-P^\tau(\hat \s',\cdot)} & \leq \tv{P^\tau(\hat \s,\cdot)-\hat P^\tau(\hat \s,\cdot)} + \tv{\hat P^\tau(\hat \s,\cdot)- \hat \pi}\\
 & \quad + \tv{\hat \pi - \hat P^\tau(\hat \s',\cdot)} + \tv{\hat P^\tau(\hat \s',\cdot)- P^\tau(\hat \s',\cdot)}\\
 & \leq 4 \cdot \frac{\delta}{8} = \delta/2,
\end{align*}
where the last inequality is due to Lemma~\ref{lemma:restricted} by taking $\beta$ sufficiently large (note that $\eta p \tau = \eta p_\beta t^{(\beta)}_{mix}(\delta/8)\leq\eta p_\beta \hat \tau_\beta \log (8/\delta)$, which tends to $0$ as $\beta \rightarrow \infty$ by hypothesis).
Consider an interval $T$ of length $R \cdot \left\lceil \frac{\log (8\ell/\delta)}{\log(1/\varepsilon)}\right\rceil$. By applying Lemma~\ref{le:hit:general} with $k=\ell$, $(R,\varepsilon) = (T,\delta/4\ell)$ and $\beta$ sufficiently large we have that for every $\s \in G$
 $$
  \Prob{\s}{X_{\ell T} \notin \hat G} \leq \delta/8.
 $$
Let $t^\star = \ell T + \tau$ and $Q = P^{\ell T}$. Then, for every $\s,\s' \in G$
\begin{align*}
 \tv{\pi - P^{t^\star}(\s', \cdot)} & \leq \tv{P^{t^\star}(\s, \cdot) - P^{t^\star}(\s', \cdot)} = \tv{Q(\s, \cdot)P^\tau - Q(\s', \cdot)P^\tau}\\
 \mbox{(triangle inequality)} & \leq \tv{Q(\s, \cdot)P^\tau - \hat Q(\s, \cdot)P^\tau} + \tv{\hat Q(\s, \cdot)P^\tau - \hat Q(\s', \cdot)P^\tau}\\
 & \quad + \tv{\hat Q(\s', \cdot)P^\tau - Q(\s', \cdot)P^\tau},
\end{align*}
where, for every $\s, \s' \in G$, we set
$$
 \hat Q(\s,\s') = \begin{cases}
                           \frac{Q(\s,\s')}{Q(\s,\hat G)}, & \mbox{if } \s, \s' \in \hat G;\\
                           0, & \mbox{otherwise.}
                          \end{cases}
$$
By \eqref{eq:samePin_mu} we obtain
$$
 \tv{Q(\s, \cdot)P^\tau - \hat Q(\s, \cdot)P^\tau} \leq \tv{Q(\s, \cdot) - \hat Q(\s, \cdot)} \leq \Prob{\s}{X_{\ell T} \notin \hat G} \leq \delta/8.
$$
By \eqref{eq:samePinP} we obtain
$$
 \tv{\hat Q(\s, \cdot)P^\tau - \hat Q(\s', \cdot)P^\tau} \leq \max_{\hat \s, \hat \s' \in \hat G} \tv{P^\tau(\hat \s, \cdot) - P^\tau(\hat \s', \cdot)} \leq \delta/2
$$
and thus $\tv{\pi - P^{t^\star}(\s', \cdot)} \leq 3\delta/4$.
Finally, for every $\hat \s \in \hat G$, by triangle inequality
\begin{align*}
 \tv{\pi - \hat \pi} & \leq \tv{\pi - P^{t^\star}(\hat \s, \cdot)} + \tv{P^{t^\star}(\hat \s, \cdot) - \hat P^{t^\star}(\hat \s, \cdot)} + \tv{\hat P^{t^\star}(\hat \s, \cdot) - \hat \pi}\\
 & \leq 3 \delta/4 + \delta/8 + \delta/8 = \delta.\tag*{\qed}
\end{align*}
\let\qed\relax
\end{proof}

\end{document}